\documentclass[letterpaper, 10 pt, conference]{ieeeconf}

\IEEEoverridecommandlockouts
\usepackage[utf8]{inputenc}
\usepackage{amsmath,amssymb,amsfonts,mathtools}

\usepackage{amsthm}
\usepackage{xcolor}
\usepackage{graphicx}
\usepackage{float}
\usepackage{babel}
\usepackage{algorithm}
\usepackage{algorithmic}
\usepackage{bbm}
\usepackage{booktabs}
\usepackage{caption}
\usepackage{subcaption}
\usepackage[hidelinks]{hyperref}

\newtheorem{ass}{Assumption}
\newtheorem{lemma}{Lemma}
\newtheorem{theorem}{Theorem}
\newtheorem{corollary}{Corollary}
\newtheorem{proposition}{Proposition}
\newtheorem{remark}{Remark}

\title{Residual-Conservative Model Predictive Path Integral Control}

\author{
Hyung-Jin Yoon$^{\dagger}$ and
Hunmin Kim$^{\ddagger}$
\thanks{$^{\dagger}$H.-J. Yoon is with the
Department of Mechanical and Nuclear Engineering, Tennessee Technological
University, Cookeville, TN, USA.}%
\thanks{$^{\ddagger}$H. Kim is with the School of Engineering, Department
of Electrical and Computer Engineering, Mercer University, Macon, GA, USA.}%
\thanks{This work was supported by internal funding at Tennessee
Technological University.}
}

\begin{document}
\maketitle

\begin{abstract}
Sampling-based model predictive control methods handle nonlinear
dynamics and complex cost landscapes through Monte Carlo rollouts,
yet typically employ fixed constraint penalties that do not adapt to
model-plant mismatch. This paper proposes RC-MPPI, a sampling-based
MPC framework that modulates safety conservatism online via three
coupled mechanisms: residual-dependent constraint tightening, adaptive
safety-cost shaping, and residual-adaptive temperature relaxation.
The temperature adaptation reflects a key insight: when the model is
inaccurate, rollout cost evaluations become unreliable, and raising
temperature reduces overcommitment to apparent cost rankings. As
model-plant mismatch increases, safety margins tighten and the MPPI
temperature rises; as prediction fidelity improves, nominal MPPI
behavior is recovered. Under Lipschitz dynamics and sub-Gaussian
disturbances, we derive probabilistic bounds on constraint violation
and establish that the joint effect of all three mechanisms
monotonically reduces violation probability with growing residual. A
rollout-cost uncertainty analysis further shows that mismatch-induced
perturbations of MPPI importance weights scale proportionally to
residual magnitude and inversely with temperature, providing
theoretical justification for residual-adaptive temperature relaxation
under model-plant mismatch. We extend the method to a
two-time-scale architecture with episodic model refinement.
Simulations on an LTI point-mass and a planar 2R manipulator confirm
systematic improvements in safety margin, success rate, and control
efficiency over vanilla MPPI.
\end{abstract}

\section{Introduction}
Low-cost robotic and embedded control platforms frequently exhibit
execution variability arising from actuation lag, saturation, unmodeled
inner-loop dynamics, and limited sensing fidelity. When high-level
planners rely on simplified or nominal actuator models, such effects
induce model-plant mismatch that can compromise constraint
satisfaction. In receding-horizon implementations, this mismatch is
persistent and state-dependent rather than a one-time disturbance.

Sampling-based model predictive control methods, such as Model
Predictive Path Integral (MPPI)
control~\cite{williams2015mppi,williams2017tor}, handle nonlinear
dynamics and complex cost landscapes through Monte Carlo rollouts and
importance weighting. However, standard formulations assume a fixed
nominal model and employ static constraint penalties or barrier
functions that do not adapt to model-plant mismatch. When
model-plant mismatch grows, fixed safety mechanisms may become
insufficient.

This paper addresses the problem of online conservatism adaptation
under model-plant mismatch. Rather than performing real-time parameter
identification or maintaining a full disturbance belief, we exploit
the prediction--execution residual, the directly measurable discrepancy
between predicted and realized state transitions, as a lightweight
signal of model-plant mismatch, and embed it into the sampling-based
MPC objective via three coupled adaptive mechanisms.

The resulting method, Residual-Conservative MPPI (RC-MPPI), enforces
safety through residual-dependent constraint tightening and amplified
penalty scaling, while the MPPI temperature is \emph{relaxed} in
proportion to the observed prediction--execution residual. Standard
MPPI treats temperature as a fixed exploration parameter; under
model-plant mismatch, however, this conflates genuine cost differences
with artifacts of an inaccurate model. RC-MPPI instead treats
temperature as an epistemic parameter encoding confidence in rollout
cost evaluations. We show that the $\ell_1$ deviation between true
and nominal importance weights is bounded by $2C_\Delta \bar{s}_k /
\beta_k$, where $C_\Delta$ captures cost sensitivity to model-plant
mismatch and $\bar{s}_k$ is the filtered residual. Since this bound
grows with mismatch and shrinks with temperature, raising $\beta_k$
as $\bar{s}_k$ increases directly limits the distortion of MPPI
weights by corrupted rollout cost rankings.

\textbf{Contributions.}
\begin{itemize}
\item RC-MPPI, a sampling-based MPC scheme with three coupled
residual-adaptive mechanisms: constraint tightening, penalty scaling,
and residual-adaptive sampling modulation comprising temperature
relaxation and exploration contraction, with no additional
computational cost over standard MPPI. Temperature is treated as an
epistemic parameter encoding confidence in rollout cost evaluations
under model-plant mismatch.
\item A probabilistic safety analysis grounded in an $N$-step horizon
prediction error bound (Theorem~\ref{thm:horizon_bound}), establishing
that the joint effect of all three mechanisms monotonically reduces
constraint violation probability with growing residual
(Proposition~\ref{prop:joint_safety}).
\item A rollout-cost uncertainty analysis
(Proposition~\ref{prop:temp_sensitivity}) showing that
mismatch-induced perturbations of MPPI importance weights are bounded
by $2C_\Delta \bar{s}_k / \beta_k$, providing theoretical
justification for raising temperature as model-plant mismatch grows.
\end{itemize}

\section{Related Work}
Model Predictive Path Integral (MPPI) control is a sampling-based
approximation of stochastic optimal control derived from
information-theoretic
principles~\cite{williams2015mppi,williams2017tor}. By performing
importance sampling over trajectory rollouts and computing a soft-min
update, MPPI enables real-time nonlinear control without explicit
gradient computation. In baseline implementations, constraints are
handled through soft penalties, barrier-like shaping, or rollout
truncation; however, hard rejection of infeasible trajectories reduces
the effective sample size, leading to weight degeneracy. To address
safety requirements, recent work integrates Control Barrier
Functions~\cite{ames2017cbf} into MPPI. Shield-MPPI introduces a
CBF-inspired shielding mechanism~\cite{yu2023shieldmppi}, and
guaranteed-safe MPPI variants employ composite CBF
constructions~\cite{rabiee2025guaranteed}. Other approaches propagate
uncertainty explicitly: unscented MPPI uses sigma-point
approximations~\cite{unscentedmppi2023}, and belief-space stochastic
MPPI enforces approximate chance constraints~\cite{bssmppi2024}.

More broadly, chance-constrained MPC enforces probabilistic constraint
satisfaction~\cite{mesbah2016smc,blackmore2011chance}, and classical
risk-sensitive control penalizes tail events through exponential cost
transformations~\cite{whittle1990risk}. Robust MPC enforces safety
through tube-based tightening~\cite{rawlings2017mpcbook}, and
learning-based MPC incorporates model adaptation while preserving
constraint satisfaction~\cite{aswani2013lbmpc}. RC-MPPI complements
these directions by modulating conservatism using the filtered
prediction--execution residual. A distinguishing feature is the
treatment of temperature as an \emph{epistemic} parameter encoding
confidence in rollout cost evaluations, supported by a rollout-cost
uncertainty analysis that links model-plant mismatch,
importance-weight sensitivity, and adaptive temperature selection.

\section{System Model}\label{sec:system_model}
Consider a discrete-time system with state $\mathbf{x}_k \in
\mathcal{X} \subseteq \mathbb{R}^{n_x}$ and control $\mathbf{u}_k
\in \mathcal{U} \subseteq \mathbb{R}^{n_u}$, where $\mathcal{X}$ and
$\mathcal{U}$ are compact sets. The controller plans using a nominal
parametric predictor
\begin{equation}\label{eq:nominal_dynamics}
\mathbf{x}_{k+1} = f_\theta(\mathbf{x}_k, \mathbf{u}_k),
\end{equation}
where $\theta$ denotes the nominal model parameter. The true system
evolves as $\mathbf{x}_{k+1} = \tilde{f}(\mathbf{x}_k, \mathbf{u}_k)
+ \mathbf{w}_k$, where $\tilde{f}$ is the true (unknown) dynamics and
$\mathbf{w}_k \in \mathbb{R}^{n_x}$ is a stochastic process
disturbance. The controller receives noisy state measurements
$\mathbf{y}_k = \mathbf{x}_k + \mathbf{v}_k$, where $\mathbf{v}_k
\in \mathbb{R}^{n_x}$ is the measurement noise. At each time step
$k$, the controller optimizes over a finite planning horizon of
$N \geq 1$ steps. Safety is encoded by a constraint function
$h : \mathcal{X} \rightarrow \mathbb{R}$, where the safe set is
$\mathcal{X}_s := \{\mathbf{x} \in \mathcal{X} \mid h(\mathbf{x})
\leq 0\}$ and $h$ is Lipschitz with constant $L_h > 0$.

\begin{ass}[Local Lipschitz Nominal Dynamics]\label{ass:lipschitz}
There exists $L_f > 0$ such that for all $(\mathbf{x}, \mathbf{u}),
(\mathbf{x}', \mathbf{u})$ in $\mathcal{X} \times \mathcal{U}$,
$\|f_\theta(\mathbf{x}, \mathbf{u}) - f_\theta(\mathbf{x}',
\mathbf{u})\| \leq L_f \|\mathbf{x} - \mathbf{x}'\|$.
\end{ass}

\begin{ass}[Sub-Gaussian Disturbance]\label{ass:subgaussian}
$\mathbf{w}_k$ is i.i.d.\ sub-Gaussian with scalar variance proxy
$\sigma_x^2 > 0$: there exist $a_1, a_2 > 0$ such that
$\mathbb{P}(\|\mathbf{w}_k\| \geq t) \leq a_1 \exp(-a_2 t^2 /
\sigma_x^2)$ for all $t \geq 0$. Furthermore, $\mathbf{w}_k$ is
independent of the natural filtration $\mathcal{F}_k :=
\sigma(\mathbf{x}_0, \mathbf{w}_0, \mathbf{v}_0, \ldots,
\mathbf{w}_{k-1}, \mathbf{v}_{k-1})$ generated by the system history
up to time $k$.
\end{ass}

\begin{ass}[Nominal MPPI Competence]\label{ass:nominal_competence}
Under zero prediction--execution residual ($\bar{s}_k = 0$, as
defined in~\eqref{eqn:fil_statistics}), vanilla MPPI with nominal
temperature $\beta_0$ achieves constraint satisfaction with
probability at least $1 - \delta_0$ for some $\delta_0 \in (0, 1)$,
and the prior mean control sequence $\bar{\mathbf{U}}$ (initialized
from the shifted solution of the previous time step) lies in the safe
region with probability at least $1 - \delta_0$.
\end{ass}

\begin{remark}[Practical validity of
Assumption~\ref{ass:nominal_competence}]
Assumption~\ref{ass:nominal_competence} is behavioral, not structural:
it requires no convexity or geometric property of the cost landscape,
only that the baseline works when the model is accurate. This is
consistent with the simulation results in
Section~\ref{sec:simulation}, where vanilla MPPI performs competently
under moderate mismatch conditions.
\end{remark}

\begin{ass}[Bounded Measurement Noise]\label{ass:bounded_noise}
The measurement noise satisfies $\|\mathbf{v}_k\| \leq \bar{v}$
almost surely for some $\bar{v} > 0$.
\end{ass}

\begin{remark}[Practical validity of
Assumption~\ref{ass:bounded_noise}]
Assumption~\ref{ass:bounded_noise} is satisfied by sensors with
finite resolution or hardware-level clipping, which is standard on
low-cost robotic and embedded platforms. It separates measurement
noise, which is bounded by sensor characteristics, from process
disturbance $\mathbf{w}_k$, which is modeled as sub-Gaussian to
capture environment uncertainty.
\end{remark}

\subsection{Residual Estimation}
After executing $\mathbf{u}_{k-1}$ and measuring $\mathbf{y}_k =
\mathbf{x}_k + \mathbf{v}_k$, the one-step prediction residual is
\begin{equation}\label{eq:residual_def}
\mathbf{r}_k := \mathbf{y}_k - f_\theta(\mathbf{y}_{k-1},
\mathbf{u}_{k-1}).
\end{equation}
A scalar mismatch indicator $s_k := \|\mathbf{W}_r \mathbf{r}_k\|$
(invertible weighting matrix $\mathbf{W}_r$) is filtered as
\begin{equation}\label{eqn:fil_statistics}
\bar{s}_k = (1-\rho)\bar{s}_{k-1} + \rho s_k.
\end{equation}
Define the aggregated horizon disturbance
\begin{equation}\label{eq:xi_def}
\boldsymbol{\xi}_{k+N} := \sum_{t=0}^{N-1}
L_f^{N-1-t}\mathbf{w}_{k+t},
\end{equation}
which accumulates the stochastic disturbances over the planning
horizon weighted by the Lipschitz expansion factor $L_f$. Since
$\mathbf{r}_k$ depends only on $\mathbf{y}_k$, $\mathbf{y}_{k-1}$,
and $\mathbf{u}_{k-1}$, the filtered statistic $\bar{s}_k$ is
$\mathcal{F}_{k+1}$-measurable. By
Assumption~\ref{ass:subgaussian}, the future disturbances
$\mathbf{w}_{k}, \ldots, \mathbf{w}_{k+N-1}$ are independent of
$\mathcal{F}_{k+1}$, so $\boldsymbol{\xi}_{k+N}$ is independent of
$\bar{s}_k$ given $\mathcal{F}_{k+1}$. This separation is the key
that allows the deterministic residual-dependent term $c_r \bar{s}_k
+ c_0$ and the stochastic term $\|\boldsymbol{\xi}_{k+N}\|$ to be
bounded independently in Theorem~\ref{thm:horizon_bound}.

\begin{theorem}[Implementable Horizon Prediction Error Bound]
\label{thm:horizon_bound}
Under Assumptions~\ref{ass:lipschitz}--\ref{ass:bounded_noise},
conditioned on $\mathcal{F}_{k+1}$, the terminal prediction error
$\mathbf{e}_{k+N} := \mathbf{x}_{k+N} - \hat{\mathbf{x}}_{k+N}$
satisfies
\begin{equation}\label{eq:corollary_bound}
\|\mathbf{e}_{k+N}\| \leq c_r \bar{s}_k + c_0 +
\|\boldsymbol{\xi}_{k+N}\|,
\end{equation}
where
\[
c_r = S_N \|\mathbf{W}_r^{-1}\|/\rho,
\qquad
c_0 = (L_f^N + S_N)\bar{v},
\]
are $\mathcal{F}_{k+1}$-measurable planning-time constants, and
$\boldsymbol{\xi}_{k+N}$ is independent of $\mathcal{F}_{k+1}$ with
sub-Gaussian variance proxy $\sigma_x^2(L_f^{2N}-1)/(L_f^2-1)$.
\end{theorem}
\begin{proof}
\textit{Step~1 (One-step residual-to-deviation).}
The one-step prediction residual satisfies
$\mathbf{r}_k = (\mathbf{x}_k - \hat{\mathbf{x}}_{k|k-1}) +
\mathbf{v}_k$, so
\[
\|\mathbf{x}_k - \hat{\mathbf{x}}_{k|k-1}\|
\leq \|\mathbf{r}_k\| + \bar{v}
\leq \|\mathbf{W}_r^{-1}\|s_k + \bar{v}.
\]
The filter recursion~\eqref{eqn:fil_statistics} gives $\bar{s}_k
\geq \rho s_k$, hence $s_k \leq \rho^{-1}\bar{s}_k$, yielding
\begin{equation}\label{eq:one_step_bound}
\|\mathbf{x}_k - \hat{\mathbf{x}}_{k|k-1}\|
\leq \tilde{c}_r\bar{s}_k + \tilde{c}_0,
\quad
\tilde{c}_r := \|\mathbf{W}_r^{-1}\|/\rho,\;
\tilde{c}_0 := \bar{v}.
\end{equation}

\textit{Step~2 ($N$-step error propagation).}
Let $\boldsymbol{\delta}_t := \mathbf{x}_{k+t} -
\hat{\mathbf{x}}_{k+t}$ with $\|\boldsymbol{\delta}_0\| \leq
\bar{v}$ (since $\hat{\mathbf{x}}_k = \mathbf{y}_k =
\mathbf{x}_k + \mathbf{v}_k$). The recursion
\[
\boldsymbol{\delta}_{t+1}
= \underbrace{f_\theta(\mathbf{x}_{k+t},\mathbf{u}_{k+t})
- f_\theta(\hat{\mathbf{x}}_{k+t},\mathbf{u}_{k+t})}_{\leq
L_f\|\boldsymbol{\delta}_t\|}
+ \boldsymbol{\Delta}_{k+t} + \mathbf{w}_{k+t}
\]
with $\|\boldsymbol{\Delta}_{k+t}\| \leq \|\mathbf{W}_r^{-1}\|
s_{k+t+1} + \bar{v}$ gives, upon unrolling over $N$ steps,
\[
\|\boldsymbol{\delta}_N\|
\leq
S_N\|\mathbf{W}_r^{-1}\|\bar{s}_k^N
+ (L_f^N + S_N)\bar{v}
+ \|\boldsymbol{\xi}_{k+N}\|,
\]
where $\bar{s}_k^N := \max_{0\leq t\leq N-1}s_{k+t+1}$ and
$S_N = (L_f^N-1)/(L_f-1)$ for $L_f>1$ (or $N$ for $L_f=1$).

\textit{Step~3 (Replacing $\bar{s}_k^N$ with $\bar{s}_k$).}
Under the planning-window stationarity condition, the mismatch level
does not increase over the horizon $[k, k+N-1]$, so
$\bar{s}_k^N \leq s_k$ almost surely. The filter
recursion~\eqref{eqn:fil_statistics} gives $\bar{s}_k \geq \rho s_k$,
hence $s_k \leq \rho^{-1}\bar{s}_k$. Chaining these two inequalities
yields
\[
\bar{s}_k^N \leq s_k \leq \rho^{-1}\bar{s}_k,
\]
so substituting gives the $\mathcal{F}_{k+1}$-measurable
bound~\eqref{eq:corollary_bound} with $c_r =
S_N\|\mathbf{W}_r^{-1}\|/\rho$ and $c_0 = (L_f^N + S_N)\bar{v}$.

\textit{Step~4 (Sub-Gaussian tail of $\boldsymbol{\xi}_{k+N}$).}
Since $\boldsymbol{\xi}_{k+N} = \sum_{t=0}^{N-1}
L_f^{N-1-t}\mathbf{w}_{k+t}$ is a weighted sum of independent
sub-Gaussian disturbances with squared coefficient sum
$\sum_{t=0}^{N-1}L_f^{2t} = (L_f^{2N}-1)/(L_f^2-1)$, it is
sub-Gaussian with the stated variance proxy. Independence from
$\mathcal{F}_{k+1}$ follows from Assumption~\ref{ass:subgaussian}.
\end{proof}

\section{Residual-Conservative MPPI}
RC-MPPI augments vanilla MPPI with three residual-driven mechanisms
that tighten safety and reduce model-confidence as execution mismatch
grows. Algorithm~\ref{alg:rcmppi} summarizes the per-step procedure.

\begin{algorithm}[H]
\caption{Residual-Conservative MPPI (RC-MPPI)}
\label{alg:rcmppi}
\begin{algorithmic}[1]
\STATE Measure $\mathbf{y}_k$; compute $\mathbf{r}_k$
       via~\eqref{eq:residual_def}; update $\bar{s}_k$
       via~\eqref{eqn:fil_statistics}
\STATE Compute $m(\bar{s}_k)$, $\alpha_k$, $\varsigma_k$, $\beta_k$
       via~\eqref{eq:tightening_margin}--\eqref{eq:sampling_modulation}
\STATE Sample $\boldsymbol{\epsilon}^{(i)} \sim \mathcal{N}(0,
       \varsigma_k^2 \mathbf{I})$; roll out~\eqref{eq:nominal_dynamics}
\STATE Evaluate $Z^{(i)}$ with barrier
       $\alpha_k \phi(h(\mathbf{x}) + m(\bar{s}_k))$
\STATE Compute weights~\eqref{eq:mppi_weights} at temperature $\beta_k$;
       update $\mathbf{u}_k$ via~\eqref{eq:mppi_update}; execute
       $\mathbf{u}_k$
\end{algorithmic}
\end{algorithm}

Step~1 converts the raw measurement into a scalar mismatch signal
$\bar{s}_k$. Step~2 translates that signal into three coupled
residual-adaptive mechanisms: (i) residual-dependent barrier
tightening $m(\bar{s}_k)$ that deforms the effective safe set;
(ii) amplified barrier penalty $\alpha_k$ that strengthens the cost
signal against unsafe rollouts; and (iii) residual-adaptive
temperature relaxation $(\varsigma_k, \beta_k)$ that contracts
exploration and softens importance weights to reflect reduced
confidence in cost evaluations under an inaccurate model.
Steps~3--5 are standard MPPI rollout and update, now operating under
the residual-adapted cost and sampling distribution. When $\bar{s}_k
\approx 0$ all three modulations vanish and RC-MPPI reduces to
vanilla MPPI.

\subsection{Risk-Sensitive MPPI}

MPPI minimizes the entropic (free-energy) cost functional
\begin{equation}\label{eq:entropic_obj}
J_\beta(\mathbf{U})
=-\beta\log\mathbb{E}\!\left[
\exp\!\left(-\frac{Z(\mathbf{U})}{\beta}\right)
\right],
\end{equation}
where $\beta > 0$ is the \emph{temperature}. Larger $\beta$ reduces
sensitivity to cost differences across rollouts, producing a more
uniform weight distribution; smaller $\beta$ sharpens the weights
onto the lowest-cost rollout.

In RC-MPPI, $\beta$ encodes \emph{confidence in rollout cost
evaluations}. When the model is accurate, small $\beta$ concentrates
mass on the genuinely best rollout. When the model is wrong, cost
evaluations are unreliable, and large $\beta$ prevents
over-commitment to a rollout that merely \emph{appears} optimal under
the mismatched model. This interpretation is formalized in
Proposition~\ref{prop:temp_sensitivity}, which shows that the
sensitivity of importance weights to rollout-cost uncertainty scales
as $2C_\Delta \bar{s}_k / \beta_k$, providing theoretical
justification for raising $\beta_k$ as model-plant mismatch grows.

\paragraph{Trajectory cost.}
\begin{equation}\label{eq:traj_cost}
\begin{aligned}
Z(\mathbf{U}) = &\sum_{t=0}^{N-1}\!\Big(
\ell_{\mathrm{trk}}(\mathbf{x}_{k+t})
+ \ell_u(\mathbf{u}_{k+t})
+ \ell_{\mathrm{safe}}(\mathbf{x}_{k+t};\bar{s}_k)
\Big) \\
&+ \ell_f(\mathbf{x}_{k+N}),
\end{aligned}
\end{equation}
where $\ell_{\mathrm{trk}}(\mathbf{x}) = \|\mathbf{x} -
\mathbf{x}^{\mathrm{ref}}\|_{\mathbf{Q}}^2$ with positive definite
weight matrices $\mathbf{Q}, \mathbf{Q}_f \succ 0$,
$\ell_u(\mathbf{u}) = \|\mathbf{u}\|_{\mathbf{R}}^2$ with $\mathbf{R}
\succ 0$,
\begin{equation}\label{eq:barrier_cost}
\ell_{\mathrm{safe}}(\mathbf{x};\bar{s}_k)
= \alpha_k\,\phi\!\left(h(\mathbf{x}) + m(\bar{s}_k)\right),
\end{equation}
$\phi(z) = \max(0,z)^2$, $m(\bar{s}_k)$ is the residual-dependent
tightening margin defined in~\eqref{eq:tightening_margin}, and
$\ell_f(\mathbf{x}) = \|\mathbf{x} -
\mathbf{x}^{\mathrm{ref}}_{k+N}\|_{\mathbf{Q}_f}^2$.

\paragraph{Importance sampling update.}
Perturbations $\boldsymbol{\epsilon}^{(i)} \sim \mathcal{N}(0,
\varsigma_k^2 \mathbf{I})$ form sequences $\mathbf{U}^{(i)} =
\mathbf{U} + \boldsymbol{\epsilon}^{(i)}$ with costs $Z^{(i)} :=
Z(\mathbf{U}^{(i)})$.
\begin{align}
w^{(i)} &= \frac{\exp(-Z^{(i)}/\beta_k)}
               {\sum_j \exp(-Z^{(j)}/\beta_k)},
\label{eq:mppi_weights}\\
\mathbf{u}_k &\leftarrow \mathbf{u}_k
+ \sum_{i=1}^K w^{(i)} \boldsymbol{\epsilon}^{(i)}_0.
\label{eq:mppi_update}
\end{align}

\subsection{Residual-Conservative Barrier Modulation}

\paragraph{Residual-dependent tightening.}
\begin{equation}\label{eq:tightening_margin}
m(\bar{s}_k) := L_h(c_r \bar{s}_k + c_0),
\end{equation}
where $c_r$ and $c_0$ are the explicit horizon-dependent constants
from Theorem~\ref{thm:horizon_bound}: $c_r = S_N
\|\mathbf{W}_r^{-1}\|/\rho$ and $c_0 = (L_f^N + S_N)\bar{v}$. The
margin thus scales with horizon $N$ and Lipschitz constant $L_f$:
longer horizons or more expansive dynamics require greater tightening,
as expected.

\begin{remark}[Implementation]
For $h(\mathbf{x}) = r - d(\mathbf{x})$, radius inflation
$r_{\mathrm{eff}} = r + \mathrm{clip}(\kappa_r \bar{s}_k, 0,
\Delta r_{\max})$ implements tightening since $h(\mathbf{x}) +
m(\bar{s}_k) = (r + m(\bar{s}_k)) - d(\mathbf{x})$, where
$\mathrm{clip}(v, 0, \Delta r_{\max}) := \min(\max(v, 0),
\Delta r_{\max})$ saturates the inflation to the interval
$[0, \Delta r_{\max}]$.
\end{remark}

\paragraph{Residual-aware penalty scaling.}
\begin{equation}\label{eq:penalty_scaling}
\alpha_k = \alpha_0(1 + \gamma \bar{s}_k).
\end{equation}

\paragraph{Residual-adaptive sampling modulation.}
\begin{align}\label{eq:sampling_modulation}
\varsigma_k &= \mathrm{clip}\!\left(
  \frac{\varsigma_0}{1 + \kappa_\varsigma \bar{s}_k},
  \varsigma_{\min}, \varsigma_{\max}\right), \nonumber\\
\beta_k &= \mathrm{clip}\!\left(
  \beta_0(1 + \kappa_\beta \bar{s}_k),
  \beta_{\min}, \beta_{\max}\right).
\end{align}
As $\bar{s}_k$ increases: $\varsigma_k \downarrow$ contracts
exploration; $\beta_k \uparrow$ softens importance weights to reflect
reduced model confidence.

\begin{remark}[Complementary roles of $\alpha_k$ and $\beta_k$]
$\alpha_k$ amplifies the barrier signal so unsafe rollouts incur
large cost regardless of temperature. $\beta_k$ controls exploitation
of that signal. Although $\beta_k \uparrow$ softens weights, the
barrier cost $\alpha_k \phi(m(\bar{s}_k))$ grows as
$O(\bar{s}_k^2)$ while $\beta_k$ grows as $O(\bar{s}_k)$: their
ratio diverges, and unsafe rollouts receive asymptotically zero
weight despite the rising temperature
(Lemma~\ref{lem:temperature_averaging}(ii)). The performance
advantage of $\beta \uparrow$ over $\beta \downarrow$ under mismatch
is established in Proposition~\ref{prop:temp_sensitivity}.
\end{remark}

\section{Residual-Adaptive Safety Analysis}

Terminal constraint satisfaction is a standard analysis paradigm in
receding-horizon control~\cite{rawlings2017mpcbook}, and
probabilistic formulations are well-established in stochastic and
chance-constrained
MPC~\cite{mesbah2016smc,blackmore2011chance}. We adopt this framework
to bound the probability that the nominal terminal prediction violates
the safety constraint $h(\mathbf{x}_{k+N}) \leq 0$, and show that
all three RC-MPPI mechanisms jointly reduce this probability as
model-plant mismatch grows. Proposition~\ref{prop:terminal_safety_bound}
establishes a baseline from constraint tightening.
Lemma~\ref{lem:temperature_averaging} and
Lemma~\ref{lem:sampling_concentration} characterize temperature
adaptation and trajectory concentration.
Proposition~\ref{prop:joint_safety} gives the joint bound, and
Corollary~\ref{cor:dominance} establishes that RC-MPPI achieves at
least the constraint satisfaction probability of vanilla MPPI, with
strict improvement whenever $\bar{s}_k > 0$. Separately,
Proposition~\ref{prop:temp_sensitivity} establishes that raising
$\beta_k$ reduces the sensitivity of importance weights to
rollout-cost uncertainty induced by model-plant mismatch, providing
theoretical justification for the residual-adaptive temperature rule.

\begin{proposition}[Terminal Safety Bound]
\label{prop:terminal_safety_bound}
Under Assumptions~\ref{ass:lipschitz}--\ref{ass:subgaussian}, let
$d_{\mathrm{safe}} := -h(\hat{\mathbf{x}}_{k+N}) > 0$ denote the
nominal terminal safety margin. Applying
Theorem~\ref{thm:horizon_bound} to bound the terminal prediction
error, if $d_{\mathrm{safe}} \geq m(\bar{s}_k)$, then
\begin{equation}\label{eq:terminal_prob_bound}
\mathbb{P}(h(\mathbf{x}_{k+N}) > 0 \mid \mathcal{F}_{k+1})
\leq c_1 \exp\!\left(
-\frac{c_2(d_{\mathrm{safe}} - m(\bar{s}_k))^2}{\sigma_x^2}
\right),
\end{equation}
where $c_1 := a_1$ and $c_2 := a_2 / L_h^2$.
\end{proposition}
\begin{proof}
Since $h$ is Lipschitz with constant $L_h$,
\[
h(\mathbf{x}_{k+N}) \leq h(\hat{\mathbf{x}}_{k+N}) + L_h
\|\mathbf{e}_{k+N}\| = -d_{\mathrm{safe}} +
L_h\|\mathbf{e}_{k+N}\|,
\]
so $\{h(\mathbf{x}_{k+N}) > 0\} \subseteq
\{L_h\|\mathbf{e}_{k+N}\| > d_{\mathrm{safe}}\}$. By
Theorem~\ref{thm:horizon_bound}, $\|\mathbf{e}_{k+N}\|
\leq c_r\bar{s}_k + c_0 + \|\boldsymbol{\xi}_{k+N}\|$, so
multiplying by $L_h$ and using $m(\bar{s}_k) = L_h(c_r\bar{s}_k +
c_0)$ gives $L_h\|\mathbf{e}_{k+N}\| \leq m(\bar{s}_k) +
L_h\|\boldsymbol{\xi}_{k+N}\|$. Therefore,
\[
\{h(\mathbf{x}_{k+N}) > 0\} \subseteq
\{L_h\|\boldsymbol{\xi}_{k+N}\| > d_{\mathrm{safe}} -
m(\bar{s}_k)\}.
\]
Since $\boldsymbol{\xi}_{k+N}$ is independent of $\bar{s}_k$ given
$\mathcal{F}_{k+1}$ by Theorem~\ref{thm:horizon_bound},
the sub-Gaussian tail bound of Assumption~\ref{ass:subgaussian}
applies conditionally with $t = (d_{\mathrm{safe}} -
m(\bar{s}_k))/L_h$, yielding~\eqref{eq:terminal_prob_bound}.
\end{proof}

\begin{corollary}[Consistency with Nominal MPPI]
\label{cor:consistency}
The modulation functions~\eqref{eq:tightening_margin}--\eqref{eq:sampling_modulation}
are continuous in $\bar{s}_k$. Consequently,
\begin{align*}
\lim_{\bar{s}_k \to 0} m(\bar{s}_k) &= L_h c_0, &
\lim_{\bar{s}_k \to 0} \alpha_k &= \alpha_0, \\
\lim_{\bar{s}_k \to 0} \beta_k &= \beta_0, &
\lim_{\bar{s}_k \to 0} \varsigma_k &= \varsigma_0,
\end{align*}
and RC-MPPI reduces to standard MPPI with nominal temperature $\beta_0$
and fixed safety shaping.
\end{corollary}
\begin{proof}
Direct substitution of $\bar{s}_k = 0$
into~\eqref{eq:tightening_margin}--\eqref{eq:sampling_modulation},
using continuity of each expression in $\bar{s}_k$.
\end{proof}

\subsection{Joint Safety via Adaptive Weighting}

\begin{lemma}[Temperature-Averaging Under Mismatch]
\label{lem:temperature_averaging}
\smallskip\noindent\textup{(i)~[Update variance reduction.]}
The MPPI control update satisfies
\begin{equation}
\label{eq:update_variance}
\mathrm{Var}_w(\mathbf{u}_k)
\leq
\varsigma_k^2
\left(
1 + \frac{\mathrm{diam}(Z)}{\beta_k}
\right).
\end{equation}
Hence the variance bound is nonincreasing in $\beta_k$. Since
RC-MPPI increases $\beta_k$ as the residual indicator $\bar{s}_k$
grows, the control update becomes increasingly averaged, and the
variance bound tightens under larger model-plant mismatch.

\smallskip\noindent\textup{(ii)~[Barrier suppression.]}
Any unsafe rollout ($h(\mathbf{x}^{(i)}_{k+N}) > 0$) satisfies
\begin{equation}\label{eq:unsafe_weight_ratio}
\frac{w^{(i)}_{\mathrm{unsafe}}}{w^{(j)}_{\mathrm{safe}}}
\leq
\exp\!\left(
-\frac{\alpha_k \phi(m(\bar{s}_k))}{\beta_k}
\right).
\end{equation}
Since $\alpha_k \phi(m(\bar{s}_k)) \sim O(\bar{s}_k^2)$ while
$\beta_k \sim O(\bar{s}_k)$, the exponent diverges to $-\infty$:
unsafe rollouts receive asymptotically zero weight.
\end{lemma}
\begin{proof}
\textit{Part~(i).} Following the importance sampling formulation of
MPPI~\cite{williams2015mppi,williams2017tor}, let $p = \mathcal{N}(0,
\varsigma_k^2 \mathbf{I})$ denote the prior sampling distribution
over control perturbations, and let $q$ denote the
importance-weighted posterior induced by the MPPI
weights~\eqref{eq:mppi_weights}. Let $\mathbf{U} :=
\{\mathbf{u}_k, \ldots, \mathbf{u}_{k+N-1}\}$ denote the control
sequence over the planning horizon. A standard Gibbs-measure
perturbation inequality~\cite{whittle1990risk} gives
\begin{equation*}
\begin{aligned}
\mathrm{Var}_q(\mathbf{U})
&\leq
\mathrm{Var}_p(\mathbf{U})
\left(1 + \frac{\mathrm{diam}(Z)}{\beta_k}\right) \\
&=
\varsigma_k^2
\left(1 + \frac{\mathrm{diam}(Z)}{\beta_k}\right),
\end{aligned}
\end{equation*}
where $\mathrm{diam}(Z) := \max_i Z^{(i)} - \min_i Z^{(i)}$ is the
range of rollout costs. The bound is nonincreasing in $\beta_k$.
Since RC-MPPI increases $\beta_k$ as $\bar{s}_k$ grows, the variance
bound tightens under larger model-plant mismatch.

\textit{Part~(ii).} For any unsafe rollout $i$ with
$h(\mathbf{x}^{(i)}_{k+N}) > 0$, the barrier term
$\alpha_k\phi(h(\mathbf{x}^{(i)}_{k+N}) + m(\bar{s}_k))$ is
strictly positive. Since $\phi(z) = \max(0,z)^2$ and
$h(\mathbf{x}^{(i)}_{k+N}) > 0$ implies
$h(\mathbf{x}^{(i)}_{k+N}) + m(\bar{s}_k) > m(\bar{s}_k) > 0$, we
have
\[
Z^{(i)} \geq Z^{(j)} + \alpha_k\phi(m(\bar{s}_k))
\]
for any safe rollout $j$. Substituting into the weight
formula~\eqref{eq:mppi_weights} gives
\[
\frac{w^{(i)}_{\mathrm{unsafe}}}{w^{(j)}_{\mathrm{safe}}}
=
\frac{\exp(-Z^{(i)}/\beta_k)}{\exp(-Z^{(j)}/\beta_k)}
\leq
\exp\!\left(-\frac{\alpha_k\phi(m(\bar{s}_k))}{\beta_k}\right),
\]
which proves~\eqref{eq:unsafe_weight_ratio}. To show this ratio
vanishes as $\bar{s}_k$ grows, note that
\[
\phi(m(\bar{s}_k)) = L_h^2(c_r\bar{s}_k + c_0)^2 \sim
O(\bar{s}_k^2),
\]
while from~\eqref{eq:penalty_scaling}
and~\eqref{eq:sampling_modulation}, $\alpha_k \sim O(\bar{s}_k)$
and $\beta_k \sim O(\bar{s}_k)$. Therefore,
\[
\frac{\alpha_k\phi(m(\bar{s}_k))}{\beta_k}
\sim
\frac{\alpha_0\gamma L_h^2 c_r^2}{\beta_0\kappa_\beta}
\bar{s}_k^2 \rightarrow \infty \quad
\text{as } \bar{s}_k \rightarrow \infty,
\]
so the weight ratio~\eqref{eq:unsafe_weight_ratio} converges to zero
and unsafe rollouts receive asymptotically negligible weight.
\end{proof}

\begin{lemma}[Rollout Concentration]
\label{lem:sampling_concentration}
Let $\mathbf{x}_{k+t}$ denote the unperturbed nominal rollout
obtained by propagating $f_\theta$ from $\hat{\mathbf{x}}_k =
\mathbf{y}_k$ under the mean control sequence $\mathbf{U}$, and let
$\mathbf{x}^{(i)}_{k+t}$ denote the $i$-th Monte Carlo rollout
obtained by propagating $f_\theta$ under the perturbed control
sequence $\mathbf{U}^{(i)} = \mathbf{U} + \boldsymbol{\epsilon}^{(i)}$
with $\boldsymbol{\epsilon}^{(i)} \sim \mathcal{N}(0, \varsigma_k^2
\mathbf{I})$, where $\varsigma_k$ is the residual-adaptive sampling
standard deviation defined in~\eqref{eq:sampling_modulation}. Both
trajectories are initialized from $\hat{\mathbf{x}}_k$ and
propagated under the nominal model $f_\theta$ — the true plant is
not involved. This bound is used in
Proposition~\ref{prop:joint_safety} to quantify the probability that
a perturbed rollout reaches the unsafe region
$\{h(\mathbf{x}^{(i)}_{k+N}) > 0\}$. For any rollout $i$ and step
$t \geq 1$,
\begin{equation}\label{eq:concentration_bound}
\mathbb{P}\!\left(
\|\mathbf{x}^{(i)}_{k+t} - \mathbf{x}_{k+t}\| \geq \delta
\right)
\leq 2n_x \exp\!\left(
-\frac{\delta^2}{2L_f^{2t}\varsigma_k^2}
\right).
\end{equation}
Since RC-MPPI reduces $\varsigma_k$ as $\bar{s}_k$ increases
via~\eqref{eq:sampling_modulation}, the
bound~\eqref{eq:concentration_bound} tightens and the Monte Carlo
rollouts concentrate more tightly around the unperturbed nominal
rollout under larger model-plant mismatch.
\end{lemma}
\begin{proof}
Let $\boldsymbol{\delta}_t^{(i)} := \mathbf{x}_{k+t}^{(i)} -
\mathbf{x}_{k+t}$ denote the deviation between the $i$-th Monte
Carlo rollout and the unperturbed nominal rollout, both propagated
under $f_\theta$ from the same initial state $\hat{\mathbf{x}}_k =
\mathbf{y}_k$. Since both trajectories are initialized from the same
state, $\boldsymbol{\delta}_0^{(i)} = 0$. The deviation between the
perturbed and unperturbed rollouts evolves under
Assumption~\ref{ass:lipschitz} as
\[
\|\boldsymbol{\delta}_{t+1}^{(i)}\|
\leq
L_f \|\boldsymbol{\delta}_t^{(i)}\|
+
\|\boldsymbol{\epsilon}_{k+t}^{(i)}\|,
\]
where $\boldsymbol{\epsilon}_{k+t}^{(i)} \sim \mathcal{N}(0,
\varsigma_k^2 \mathbf{I})$ is the MPPI sampling perturbation applied
at step $t$ of rollout $i$. Note that no true plant dynamics appear
here — both trajectories follow $f_\theta$, so the deviation is
driven purely by the control perturbations
$\boldsymbol{\epsilon}^{(i)}$. Recursively unrolling the inequality
yields
\[
\boldsymbol{\delta}_t^{(i)}
=
\sum_{\tau=0}^{t-1}
L_f^{t-1-\tau}
\boldsymbol{\epsilon}_{k+\tau}^{(i)}.
\]
Since $\boldsymbol{\delta}_t^{(i)}$ is a linear combination of
independent Gaussian perturbations, it is itself Gaussian with
covariance bounded by
\[
\boldsymbol{\Sigma}_t
\preceq
\varsigma_k^2
\left(
\sum_{\tau=0}^{t-1}
L_f^{2(t-1-\tau)}
\right)\mathbf{I}
\leq
L_f^{2t} \varsigma_k^2 \mathbf{I}.
\]
Therefore each coordinate of $\boldsymbol{\delta}_t^{(i)}$ is
sub-Gaussian with variance proxy at most $L_f^{2t} \varsigma_k^2$.
Applying the standard sub-Gaussian tail inequality and a union bound
over the $n_x$ state coordinates gives
\[
\mathbb{P}\!\left(
\|\mathbf{x}_{k+t}^{(i)} - \mathbf{x}_{k+t}\| \geq \delta
\right)
\leq
2n_x \exp\!\left(
-\frac{\delta^2}{2L_f^{2t} \varsigma_k^2}
\right),
\]
which proves~\eqref{eq:concentration_bound}.
\end{proof}

\begin{proposition}[Joint Adaptive Safety Bound]
\label{prop:joint_safety}
Under Assumptions~\ref{ass:lipschitz}--\ref{ass:nominal_competence},
with $d_{\mathrm{safe}} \geq m(\bar{s}_k)$ and applying
Theorem~\ref{thm:horizon_bound} to bound the terminal prediction
error,
\begin{equation}
\label{eq:joint_safety_bound}
\begin{aligned}
&\mathbb{P}\!\left(
h(\mathbf{x}_{k+N})>0
\,\middle|\,
\mathcal{F}_{k+1}
\right) \\
&\quad\leq
c_1
\exp\!\left(
-\frac{
c_2\bigl(d_{\mathrm{safe}}-m(\bar{s}_k)\bigr)^2
}{
\sigma_x^2
}
\right)
\Gamma(\bar{s}_k),
\end{aligned}
\end{equation}
where
\[
\Gamma(\bar{s}_k) := \exp\!\left(
-\frac{\alpha_k\phi(m(\bar{s}_k))}{\beta_k}
-\frac{d_{\mathrm{safe}}^2}{2L_f^{2N}\varsigma_k^2}
\right) \in (0,1]
\]
is nonincreasing in $\bar{s}_k$.
\end{proposition}
\begin{proof}
\textit{Step~1.} Proposition~\ref{prop:terminal_safety_bound} gives
the baseline exponential.
\textit{Step~2.} By Lemma~\ref{lem:temperature_averaging}(ii), the
total unsafe weight is bounded by
$\exp(-\alpha_k\phi(m(\bar{s}_k))/\beta_k)$; although $\beta_k
\uparrow$, the $O(\bar{s}_k^2)$ numerator dominates the
$O(\bar{s}_k)$ denominator, so the exponent diverges to $-\infty$.
\textit{Step~3.} By Lemma~\ref{lem:sampling_concentration}, reaching
the unsafe region requires deviation $\geq d_{\mathrm{safe}}/L_f^N$,
with probability $\leq 2n_x\exp(-d_{\mathrm{safe}}^2 /
2L_f^{2N}\varsigma_k^2)$.
\textit{Step~4.} Combining Steps~2--3 via $\exp(-A)\exp(-B) =
\exp(-(A+B))$ gives $\Gamma$ (absorbing $2n_x$ into $c_1$);
multiplying with Step~1 gives~\eqref{eq:joint_safety_bound}.
Monotonicity: $\varsigma_k \downarrow$ drives the second exponent
more negative; $\alpha_k\phi(m(\bar{s}_k))/\beta_k \sim
O(\bar{s}_k^2)$ drives the first more negative. Hence $\Gamma \in
(0,1]$ nonincreasing.
\end{proof}

\begin{corollary}[RC-MPPI Dominates Vanilla MPPI Under Mismatch]
\label{cor:dominance}
Under Assumptions~\ref{ass:lipschitz}--\ref{ass:nominal_competence},
suppose $d_{\mathrm{safe}} \geq m(\bar{s}_k)$. Applying
Theorem~\ref{thm:horizon_bound} to bound the terminal prediction
error and invoking Proposition~\ref{prop:joint_safety},
\begin{equation}\label{eq:dominance_bound}
\begin{aligned}
&\mathbb{P}(h(\mathbf{x}_{k+N}) > 0 \mid \mathcal{F}_{k+1}) \\
&\quad\leq
c_1
\exp\!\left(
-\frac{c_2\bigl(d_{\mathrm{safe}}-m(\bar{s}_k)\bigr)^2}
{\sigma_x^2}
\right)
\Gamma(\bar{s}_k) \\
&\quad\leq \delta_0,
\end{aligned}
\end{equation}
so RC-MPPI achieves constraint satisfaction with probability at least
$1 - \delta_0$.
\end{corollary}
\begin{proof}
The first inequality is Proposition~\ref{prop:joint_safety}. For the
second, note that $\Gamma(\bar{s}_k) \in (0,1]$ is nonincreasing in
$\bar{s}_k$ with $\Gamma(0) = 1$. At $\bar{s}_k = 0$,
Proposition~\ref{prop:joint_safety} reduces to
Proposition~\ref{prop:terminal_safety_bound} with $m(0) = L_h c_0$,
and the bound equals $\delta_0$ by
Assumption~\ref{ass:nominal_competence}. For $\bar{s}_k > 0$,
$\Gamma(\bar{s}_k) < 1$ gives strict improvement over the vanilla
MPPI baseline.
\end{proof}
\begin{remark}[Synergistic improvement]
The leading exponential captures geometric tightening alone; $\Gamma
\leq 1$ provides additional multiplicative reduction. Critically,
rising $\beta_k$ does \emph{not} degrade safety: quadratic barrier
growth dominates linear temperature growth. As $\bar{s}_k \to 0$,
$\Gamma \to 1$ recovering Corollary~\ref{cor:consistency}.
\end{remark}

\subsection{Rollout-Cost Sensitivity Analysis}

We now prove that the $\beta \uparrow$ strategy is not merely safe
but \emph{preferable} above an explicit mismatch threshold. The key
additional ingredient is Lemma~\ref{lem:cost_perturbation}, which
bounds how model-plant mismatch perturbs the rollout cost landscape.

\begin{lemma}[Bounded Cost Perturbation]
\label{lem:cost_perturbation}
Under Assumptions~\ref{ass:lipschitz}--\ref{ass:subgaussian} and
Theorem~\ref{thm:horizon_bound}, conditioned on
$\mathcal{F}_{k+1}$, let $\mathbf{U} := \{\mathbf{u}_k, \ldots,
\mathbf{u}_{k+N-1}\} \in \mathcal{U}^N$ denote the control trajectory
over the planning horizon. The trajectory cost perturbation satisfies
\begin{equation}\label{eq:cost_perturbation}
|Z^{\mathrm{true}}(\mathbf{U})-Z^{\mathrm{nom}}(\mathbf{U})|
\leq C_\Delta\,\bar{s}_k
\end{equation}
uniformly over $\mathbf{U} \in \mathcal{U}^N$, where $C_\Delta :=
N(\alpha_k L_h L_f^N + W_{\max})(\tilde{c}_r + \tilde{c}_0)$ and
$W_{\max}$ is the maximum cost weight. Here $C_\Delta$ and $\bar{s}_k$
are both $\mathcal{F}_{k+1}$-measurable, so the bound is a
deterministic statement given the observed history.
\end{lemma}
\begin{proof}
Under Assumption~\ref{ass:lipschitz} and Theorem~\ref{thm:horizon_bound},
the true and nominal state trajectories deviate by at most
$L_f^t(\tilde{c}_r\bar{s}_k + \tilde{c}_0)$ at step $t$. Each cost
term is Lipschitz in $\mathbf{x}$ with constant bounded by
$W_{\max}(1 + \alpha_k L_h L_f^N)$. Summing over $N$ steps and
bounding $\sum_{t=0}^{N-1} L_f^t \leq NL_f^N$
gives~\eqref{eq:cost_perturbation}.
\end{proof}

\begin{proposition}[Weight Sensitivity Under Rollout-Cost Uncertainty]
\label{prop:temp_sensitivity}
Suppose the nominal rollout costs used by MPPI are perturbed by
model uncertainty as
\[
Z_i^{\mathrm{true}} = Z_i^{\mathrm{nom}} + \Delta_i,
\qquad
|\Delta_i| \leq C_\Delta \bar{s}_k,
\]
for each sampled rollout $i = 1, \ldots, K$. Let
\[
w_i(Z, \beta)
=
\frac{\exp(-Z_i/\beta)}
{\sum_{j=1}^K \exp(-Z_j/\beta)}
\]
denote the MPPI importance weight at temperature $\beta$. Then,
applying Lemma~\ref{lem:cost_perturbation}, the sensitivity of the
importance weights to rollout-cost uncertainty is bounded as
\begin{equation}
\|w(Z^{\mathrm{true}}, \beta) - w(Z^{\mathrm{nom}}, \beta)\|_1
\leq
\frac{2C_\Delta \bar{s}_k}{\beta}.
\label{eq:weight_uncertainty_bound}
\end{equation}
Consequently, increasing $\beta$ reduces the effect of model-induced
rollout-cost uncertainty on the MPPI update. Since
Lemma~\ref{lem:cost_perturbation} shows that this uncertainty grows
with the residual magnitude $\bar{s}_k$, the residual-adaptive rule
$\beta_k = \beta_0(1 + \kappa_\beta\bar{s}_k)$ reduces
overcommitment to cost rankings that become unreliable under large
model-plant mismatch.
\end{proposition}
\begin{proof}
Let $Z(\tau) = Z^{\mathrm{nom}} + \tau\Delta$ for $\tau \in [0,1]$.
By the mean-value theorem,
\[
w_i(Z^{\mathrm{true}}, \beta) - w_i(Z^{\mathrm{nom}}, \beta)
=
\int_0^1
\frac{d}{d\tau}w_i(Z(\tau), \beta)\,d\tau.
\]
For the softmax weights with negative costs,
\[
\frac{\partial w_i}{\partial Z_j}
=
-\frac{1}{\beta}w_i(\delta_{ij} - w_j).
\]
Hence
\[
\left|\frac{d}{d\tau}w_i(Z(\tau), \beta)\right|
=
\left|
\sum_{j=1}^K
\frac{\partial w_i}{\partial Z_j}\Delta_j
\right|
\leq
\frac{1}{\beta}
w_i
\sum_{j=1}^K
|\delta_{ij} - w_j|\,|\Delta_j|.
\]
Since $|\Delta_j| \leq C_\Delta\bar{s}_k$ and
\[
\sum_{j=1}^K |\delta_{ij} - w_j|
=
|1 - w_i| + \sum_{j \neq i} w_j
=
2(1 - w_i)
\leq 2,
\]
we obtain
\[
\left|\frac{d}{d\tau}w_i(Z(\tau), \beta)\right|
\leq
\frac{2C_\Delta\bar{s}_k}{\beta}w_i.
\]
Summing over $i$ and using $\sum_i w_i = 1$ gives
\begin{equation*}
\begin{aligned}
\|w(Z^{\mathrm{true}}, \beta) - w(Z^{\mathrm{nom}}, \beta)\|_1
&\leq
\int_0^1
\sum_i
\left|
\frac{d}{d\tau}w_i(Z(\tau), \beta)
\right|
d\tau \\
&\leq
\frac{2C_\Delta\bar{s}_k}{\beta}.
\end{aligned}
\end{equation*}
Thus the influence of model-induced rollout-cost errors on the MPPI
weights scales as $C_\Delta\bar{s}_k/\beta$. Increasing $\beta$
therefore reduces sensitivity to uncertain rollout rankings. Since
Lemma~\ref{lem:cost_perturbation} bounds the rollout-cost uncertainty
by a residual-dependent term, the adaptive choice $\beta_k \uparrow$
under large $\bar{s}_k$ directly implements residual-adaptive
temperature relaxation.
\end{proof}

\begin{remark}[Scope and connection to model uncertainty]
Lemma~\ref{lem:cost_perturbation} shows that the residual $\bar{s}_k$
upper-bounds the mismatch-induced rollout-cost uncertainty, and
Proposition~\ref{prop:temp_sensitivity} shows that the resulting MPPI
weight perturbation is bounded by $2C_\Delta\bar{s}_k/\beta_k$.
Thus, raising $\beta_k$ under growing mismatch reduces sensitivity to
unreliable cost rankings, while simultaneous constraint tightening and
penalty amplification maintain safety. However,
Proposition~\ref{prop:temp_sensitivity} does not advocate unbounded
temperature increase: in the limit $\beta_k \to \infty$, importance
weights become uniform and the MPPI update degenerates to unguided
random averaging. In RC-MPPI this is prevented by the clipping
in~\eqref{eq:sampling_modulation}, which keeps $\beta_k \leq
\beta_{\max}$, so temperature relaxation remains moderate and the
controller retains directional guidance from the cost landscape.
\end{remark}

\begin{remark}[Residual vs.\ stochastic disturbance]
$\bar{s}_k$ captures \emph{systematic} mismatch (actuator lag,
saturation, model error) while $\boldsymbol{\xi}_{k+N}$ models
\emph{stochastic} disturbances. RC-MPPI separates these:
$m(\bar{s}_k)$ compensates for structured bias by shrinking the
effective safe set, while the exponential
bound~\eqref{eq:terminal_prob_bound} quantifies remaining stochastic
violation probability due to $\boldsymbol{\xi}_{k+N}$.
\end{remark}
\section{Episodic Model Adaptation}

On the fast time scale, $\bar{s}_k$ drives online adaptation without
modifying $\theta$. On the slow scale, model parameters are updated
episodically, reducing prediction error and relaxing safety margins.

\begin{lemma}[Prediction Loss Reduction]
\label{lem:residual_reduction}
Let $\mathcal{L}(\theta) := \frac{1}{M}\sum_{i=1}^M
\|\mathbf{y}_i - f_\theta(\mathbf{x}_i, \mathbf{u}_i)\|^2$
denote the empirical prediction loss. If $\theta^+$ is obtained
by a descent step from $\theta$ on $\mathcal{L}$, then
$\mathcal{L}(\theta^+) \leq \mathcal{L}(\theta)$.
\end{lemma}
\begin{proof}
A descent step ensures $\mathcal{L}(\theta^+) \leq \mathcal{L}(\theta)$
by definition; $\theta$ is always a feasible starting point.
\end{proof}

\begin{theorem}[Decreasing Conservatism Under Model Improvement]
\label{thm:decreasing_conservatism}
If $\mathbb{E}\|\mathbf{r}(\theta^+)\| \leq
\mathbb{E}\|\mathbf{r}(\theta)\|$, then
$\mathbb{E}[m(\bar{s}(\theta^+))] \leq
\mathbb{E}[m(\bar{s}(\theta))]$.
\end{theorem}
\begin{proof}
Since $s_k = \|\mathbf{W}_r \mathbf{r}_k\| \leq \|\mathbf{W}_r\|
\|\mathbf{r}_k\|$, reduced $\mathbb{E}\|\mathbf{r}\|$ implies
reduced $\mathbb{E}[s_k]$. Linearity
of~\eqref{eqn:fil_statistics} then gives $\mathbb{E}[\bar{s}_t] =
(1-\rho)\mathbb{E}[\bar{s}_{t-1}] + \rho\mathbb{E}[s_t]$, so
$\mathbb{E}[\bar{s}_t]$ decreases. Since $m(\bar{s}) =
L_h(c_r\bar{s} + c_0)$ is nondecreasing in $\bar{s}$, it follows
that $\mathbb{E}[m(\bar{s}(\theta^+))] \leq
\mathbb{E}[m(\bar{s}(\theta))]$.
\end{proof}

As accuracy improves, residuals decrease, margins relax, and $\beta_k$
returns toward $\beta_0$, recovering nominal MPPI behavior.

\section{Simulation Study}\label{sec:simulation}

We evaluate RC-MPPI on two systems of increasing complexity. The
implementation is available at~\cite{yoon2025rcmppi_code}.

\subsection{LTI Point-Mass System}

\subsubsection{Setup}

The state $\mathbf{x} = [p_x, p_y, v_x, v_y]^\top$ evolves
according to the nominal discrete-time LTI model
$\mathbf{x}_{k+1} = \mathbf{A}\mathbf{x}_k + \mathbf{B}\mathbf{u}_k$
with
\begin{equation}
\mathbf{A}=
\begin{bmatrix}
1&0&\Delta t&0\\
0&1&0&\Delta t\\
0&0&1&0\\
0&0&0&1
\end{bmatrix},
\qquad
\mathbf{B}=
\begin{bmatrix}
0&0\\
0&0\\
\Delta t&0\\
0&\Delta t
\end{bmatrix},
\label{eq:lti}
\end{equation}
with $\Delta t = 0.1$\,s. The MPPI horizon is $T = 40$ with $K =
2048$ rollouts and input bound $u_{\max} = 4$. The obstacle is
centered at $\mathbf{c}^\star = [2.5, 0]^\top$ with radius $r =
1.5$\,m. The goal is $\mathbf{p}_g = [5, 0]^\top$. A trial is
considered successful if the goal is reached within $0.25$\,m and no
obstacle violation occurs.

\subsubsection{Model-Plant Mismatch}
The true plant includes a severe first-order actuator lag
\begin{equation}
\mathbf{v}^{\mathrm{act}}_{k+1}
=
(1-\alpha)\mathbf{v}^{\mathrm{act}}_k
+
\alpha\bigl(\mathbf{v}_k + \Delta t\,\mathbf{u}_k\bigr),
\end{equation}
with $\tau = 0.9$\,s and $\alpha = 1 - \exp(-\Delta t/\tau) \approx
0.105$. The planner assumes the nominal LTI model and therefore
systematically overestimates achievable velocity changes near the
obstacle, producing persistent model-plant mismatch.

\subsubsection{RC-MPPI Parameters}

Running-cost weights are $(w_g, w_v, w_u, w_T) = (5, 0.1, 0.01, 50)$
with $w_{\mathrm{obs}} = 10^4$. The residual statistic is
\begin{equation}
s_k =
\sqrt{
w_p\|\hat{\mathbf{p}}_k - \mathbf{p}_k\|^2
+
w_v'\|\hat{\mathbf{v}}_k - \mathbf{v}_k\|^2
}
\end{equation}
with $(w_p, w_v') = (1.0, 0.5)$ and filtering parameter $\rho = 0.2$.
Residual-adaptive modulation uses $\kappa_r = 1.0$, $\Delta r_{\max}
= 1.0$\,m, $\kappa_\varsigma = 0.5$, and $\kappa_\beta = 5.0$. The
MPPI temperature follows $\beta_k = \beta_0(1 + \kappa_\beta
\bar{s}_k)$ subject to clipping.

\subsubsection{Results}

We performed $n = 50$ paired-seed Monte Carlo trials of 300 control
steps each (Table~\ref{tab:mc_lti}, Fig.~\ref{fig:traj_overlay_lti}).
Vanilla MPPI achieves a success rate of $0.64$, minimum clearance
$0.05 \pm 0.10$\,m, and $4.16 \pm 6.32$ violation steps. RC-MPPI
increases the success rate to $0.94$, improves minimum clearance to
$0.13 \pm 0.09$\,m, and reduces violation steps to $0.62 \pm 2.63$.
These results indicate substantially improved safety and constraint
satisfaction under severe actuator lag and model-plant mismatch.

RC-MPPI exhibits a modest increase in time-to-goal ($232.78 \pm 22.71
\rightarrow 249.00 \pm 12.62$ steps) and path length ($15.85 \pm 0.90
\rightarrow 16.42 \pm 0.74$\,m), consistent with the intended
safety--efficiency tradeoff. As residuals increase, constraint
tightening, penalty amplification, exploration contraction, and
temperature relaxation collectively bias the controller toward safer
trajectories rather than aggressive obstacle-skimming behavior.

The representative trial (seed 21, Fig.~\ref{fig:traj_overlay_lti})
illustrates this tradeoff. Vanilla MPPI penetrates the obstacle region
(minimum clearance $-0.082$\,m, 21 violation steps), whereas RC-MPPI
maintains positive clearance ($0.193$\,m), incurs no violations, and
successfully reaches the goal.

\begin{figure}[t]
\centering
\includegraphics[width=0.90\linewidth]{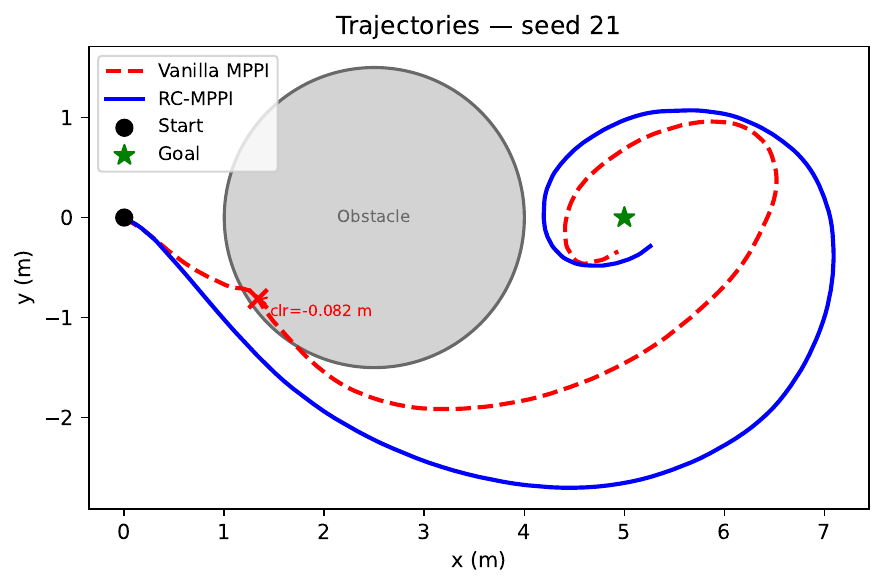}
\caption{Safety--efficiency tradeoff under severe actuator lag (seed
21, $\tau = 0.9$\,s). Vanilla MPPI (red dashed) penetrates the
obstacle and accumulates 21 violation steps, whereas RC-MPPI (blue
solid) maintains positive clearance and successfully reaches the goal
through a more conservative trajectory.}
\label{fig:traj_overlay_lti}
\end{figure}

\begin{table}[t]
\centering
\caption{LTI point-mass system ($n = 50$, $K = 2048$, $\tau =
0.9$\,s). Mean$\pm$std.}
\label{tab:mc_lti}
\begin{tabular}{lcc}
\toprule
Metric & Vanilla MPPI & RC-MPPI \\
\midrule
Success rate & 0.64 & 0.94 \\
Time-to-goal (steps) & $232.78 \pm 22.71$ & $249.00 \pm 12.62$ \\
Min clearance (m) & $0.05 \pm 0.10$ & $0.13 \pm 0.09$ \\
Violation steps & $4.16 \pm 6.32$ & $0.62 \pm 2.63$ \\
Path length (m) & $15.85 \pm 0.90$ & $16.42 \pm 0.74$ \\
\bottomrule
\end{tabular}
\end{table}

\subsection{Planar 2R Manipulator}

\subsubsection{Setup}

The true manipulator parameters are $L_1 = 1.0$\,m, $L_2 = 0.8$\,m,
$M_1 = 1.0$\,kg, and $M_2 = 0.8$\,kg. The nominal planner
intentionally uses mismatched inertial parameters
\begin{equation}
M_1^n = 1.1\ \mathrm{kg},
\qquad
M_2^n = 0.5\ \mathrm{kg}.
\end{equation}
The state is $\mathbf{x} = [q_1, q_2, \dot{q}_1, \dot{q}_2]^\top$
with bounded torque input $|u_{k,i}| \leq 6$\,N$\cdot$m ($i = 1,
2$). The nominal planner rolls out the discrete-time
Euler-integrated equations of motion for a planar (horizontal) 2R arm,
\begin{equation*}\label{eq:2r_dynamics}
\begin{bmatrix}
\mathbf{q}_{k+1} \\ \dot{\mathbf{q}}_{k+1}
\end{bmatrix}
=
\begin{bmatrix}
\mathbf{q}_k + \Delta t\,\dot{\mathbf{q}}_{k+1} \\
\dot{\mathbf{q}}_k + \Delta t\,\mathbf{M}^n(\mathbf{q}_k)^{-1}
\!\left(\mathbf{u}_k -
\mathbf{C}^n(\mathbf{q}_k,\dot{\mathbf{q}}_k)\dot{\mathbf{q}}_k
\right)
\end{bmatrix},
\end{equation*}
where $\mathbf{u}_k \in \mathbb{R}^2$ is the joint torque command,
$\mathbf{M}^n(\mathbf{q})$ is the $2 \times 2$ inertia matrix and
$\mathbf{C}^n(\mathbf{q}, \dot{\mathbf{q}})\dot{\mathbf{q}}$ is the
Coriolis/centripetal vector, both evaluated under the nominal
parameters $(M_1^n, M_2^n)$. Gravity is absent (horizontal plane).
The timestep is $\Delta t = 0.02$\,s. The goal position is
$\mathbf{p}_g = [1.35, 0.35]^\top$\,m. A circular obstacle is
located at $\mathbf{o} = [0.85, 0.20]^\top$ with radius $0.15$\,m.
The MPPI configuration uses $K = 1024$ rollouts and horizon $T = 35$.

\subsubsection{Model-Plant Mismatch}
Three independent mismatch sources are simultaneously present:
\begin{enumerate}
\item first-order actuator lag ($\tau_{\mathrm{servo}} = 0.15$\,s),
\item torque saturation at $\pm 6$\,N$\cdot$m per joint,
\item measurement noise ($\eta_q = 0.002$\,rad,
$\eta_{\dot{q}} = 0.010$\,rad/s).
\end{enumerate}
Together with the inertial parameter mismatch, these effects produce
persistent prediction--execution residuals that activate the RC-MPPI
adaptation mechanisms.

\subsubsection{Results}

We performed $n = 50$ paired-seed Monte Carlo trials of 200 control
steps each (Table~\ref{tab:2r_results}). The manipulator experiment
shows a substantial benefit from residual-adaptive conservatism.
Vanilla MPPI succeeds in 56\% of trials, whereas RC-MPPI achieves a
success rate of 96\%. Violation steps decrease from $3.08 \pm 4.41$
to $0.10 \pm 0.57$, and minimum link clearance improves from
approximately $-0.00 \pm 0.03$\,m to $0.02 \pm 0.01$\,m.

RC-MPPI also improves task efficiency in this experiment: time-to-goal
decreases from $100.56 \pm 89.47$ to $26.24 \pm 36.13$ steps,
end-effector path length decreases from $5.29 \pm 0.41$ to $3.97
\pm 0.22$\,m, and control energy decreases from $3681.03 \pm 446.17$
to $1490.88 \pm 156.49$. The representative trial
(Fig.~\ref{fig:traj_arm}) illustrates the trajectory-level safety
improvement: RC-MPPI maintains positive clearance and converges
rapidly to the goal, whereas vanilla MPPI approaches or penetrates
the obstacle boundary. Fig.~\ref{fig:clearance_arm} further shows
that RC-MPPI maintains strictly positive link clearance throughout
the trial, while vanilla MPPI crosses the zero boundary. These
results support the interpretation that residual-driven tightening,
penalty scaling, exploration contraction, and temperature relaxation
jointly prevent overcommitment to unreliable nominal rollouts under
inertial mismatch, actuator lag, saturation, and measurement noise.

\begin{figure}[t]
\centering
\includegraphics[width=0.975\linewidth]{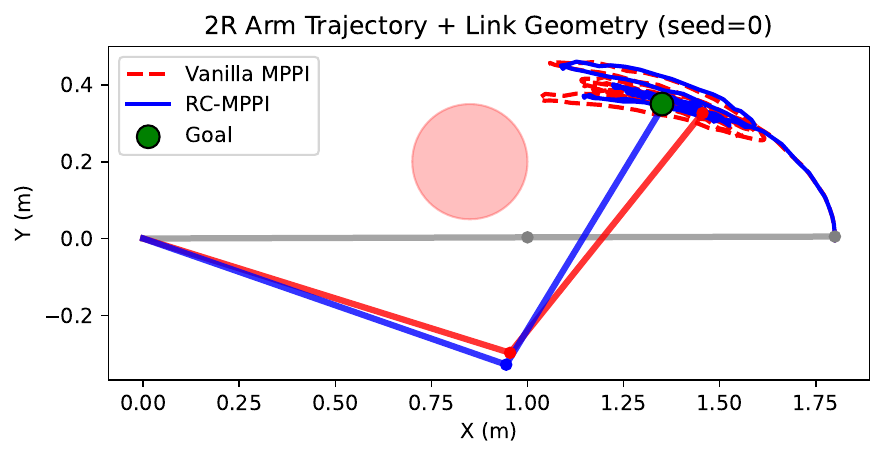}
\caption{Representative manipulator trial showing safety and
convergence improvement. Vanilla MPPI (red dashed) frequently
approaches or penetrates the obstacle boundary under model mismatch,
whereas RC-MPPI (blue solid) maintains positive clearance and
converges rapidly to the goal. Shaded disk: true obstacle.}
\label{fig:traj_arm}
\end{figure}

\begin{figure}[t]
\centering
\includegraphics[width=0.975\linewidth]{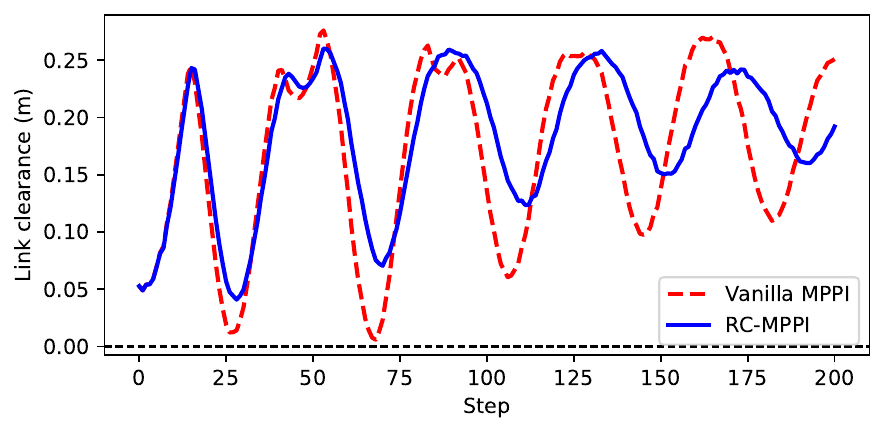}
\caption{Link clearance over time for the representative manipulator
trial (seed 0). At each step, clearance is the minimum distance from
the obstacle center to either link segment, minus the obstacle radius.
RC-MPPI (blue solid) maintains strictly positive clearance
throughout, whereas Vanilla MPPI (red dashed) approaches and crosses
the zero boundary, incurring constraint violations.}
\label{fig:clearance_arm}
\end{figure}

\begin{table}[t]
\centering
\caption{2R manipulator system ($n = 50$, $K = 1024$). Mean$\pm$std.}
\label{tab:2r_results}
\begin{tabular}{lcc}
\toprule
Metric & Vanilla MPPI & RC-MPPI \\
\midrule
Success rate & 0.56 & 0.96 \\
Time-to-goal (steps) & $100.56 \pm 89.47$ & $26.24 \pm 36.13$ \\
Min link clearance (m) & $-0.00 \pm 0.03$ & $0.02 \pm 0.01$ \\
Violation steps & $3.08 \pm 4.41$ & $0.10 \pm 0.57$ \\
EE path length (m) & $5.29 \pm 0.41$ & $3.97 \pm 0.22$ \\
Control energy & $3681.03 \pm 446.17$ & $1490.88 \pm 156.49$ \\
\bottomrule
\end{tabular}
\end{table}

\section{Conclusion}
This paper presented RC-MPPI, a residual-aware sampling-based MPC
framework that modulates safety conservatism through three coupled
mechanisms: residual-dependent constraint tightening, adaptive safety
penalty scaling, and residual-adaptive sampling modulation comprising
temperature relaxation and exploration contraction. The probabilistic
safety analysis, grounded in an $N$-step horizon prediction error
bound (Theorem~\ref{thm:horizon_bound}), establishes that the joint
effect of all three mechanisms monotonically reduces constraint
violation probability with growing residual, and that RC-MPPI
achieves at least the constraint satisfaction probability of vanilla
MPPI with strict improvement whenever model-plant mismatch is nonzero.
The rollout-cost uncertainty analysis further shows that
mismatch-induced weight perturbations are bounded by
$2C_\Delta\bar{s}_k/\beta_k$, providing theoretical justification for
treating temperature as an epistemic parameter encoding confidence in
rollout cost evaluations rather than solely as an exploration
parameter. Paired-seed Monte Carlo simulations on an LTI point-mass
system and a planar 2R manipulator confirm that RC-MPPI consistently
improves constraint satisfaction, success rate, and control efficiency
over vanilla MPPI under significant model-plant mismatch, with the
performance gap widening as model uncertainty grows. Future work will
investigate hardware validation on robotic platforms, integration with
learned residual dynamics models, and extensions to belief-space and
multi-agent sampling-based MPC formulations.

\bibliographystyle{ieeetr}
\bibliography{refs}

\end{document}